\documentclass[journal]{IEEEtran}
\ifCLASSINFOpdf
   \usepackage[pdftex]{graphicx}
\else
\fi
\usepackage{cite}

\usepackage{amsmath}

\usepackage{amsthm}

\usepackage{algorithm}

\usepackage{algpseudocode}
\usepackage{bm}

\newtheorem{define}{\textbf{Definition}}
\newtheorem{theorem}{\textbf{Theorem}}
\newtheorem{problem}{\textbf{Problem}}

\newtheorem{corollary}{\textbf{Corollary}}

\usepackage{url}

\begin{document}
%
\title{Attract More Miners to Join in Blochchain Construction for Internet of Things}
%
%
%

\author{Xingjian~Ding,
        Jianxiong~Guo,
        Deying~Li,
        and~Weili~Wu,~\IEEEmembership{Member,~IEEE}
\thanks{This work was supported by the National Natural Science Foundation of China (Grant NO.11671400, 61972404, 61672524). (\emph{Corresponding author: Deying Li.})
}
\thanks{X. Ding and D. Li are with School of Information, Renmin University of China, Beijing, 100872, China (e-mail: dxj@ruc.edu.cn; deyingli@ruc.edu.cn).}
\thanks{J. Guo and W. Wei are with the Department of Computer Science, Erik Jonsson
School of Engineering and Computer Science, University of Texas at Dallas,
Richardson, TX, 75080, USA (e-mail: jianxiong.guo@utdallas.edu; weiliwu@utdallas.edu).}
}

%
%

\markboth{Journal of \LaTeX\ Class Files,~Vol.~14, No.~8, August~2015}%
{Shell \MakeLowercase{\textit{et al.}}: Bare Demo of IEEEtran.cls for IEEE Journals}
%



\maketitle

\begin{abstract}
The world-changing blockchain technique provides a novel method to establish a secure, trusted and decentralized 
system for solving the security and personal privacy problems in Industrial Internet of Things (IIoT) applications. The mining 
process in blockchain requires miners to solve a proof-of-work puzzle, which requires high computational power. 
However, the lightweight IIoT devices cannot directly participate in the mining process due to the limitation of power
and computational resources. The edge computing service makes it possible for IIoT applications to build a blockchain
network, in which IIoT devices purchase computational resources from edge servers and thus can offload their computational
tasks. The amount of computational resource purchased by IIoT devices depends on how many profits they can get 
in the mining process, and will directly affect the security of the blockchain network. 
In this paper, we investigate the incentive 
mechanism for the blockchain platform to attract IIoT devices to purchase more computational power from edge servers
to participate in the mining process, thereby building a more secure blockchain network. We model the interaction
between the blockchain platform and IIoT devices as a two-stage Stackelberg game, where the blockchain platform act as
the leader, and IIoT devices act as followers. We analyze the existence and uniqueness of the Stackelberg equilibrium, and 
propose an efficient algorithm to compute the Stackelberg equilibrium point. Furthermore, we evaluate the performance of
our algorithm through extensive simulations, and analyze the strategies of blockchain platform and IIoT devices under
different situations.
\end{abstract}


%
\IEEEpeerreviewmaketitle


%
%
%
%

\section{introduction}
\IEEEPARstart{C}{urrently}, Internet of Things (IoT) has attracted more and more attention in many areas, such as smart cities, 
agricultures, health care, industry, etc. It is estimated that the total number of connected IoT devices will be 50 billion
by the end of 2020 \cite{IoTdevicesNum}. 
Specifically, the widespread application of the IoT has stimulated the evolution of factories to the fourth industrial revolution 
(Industry 4.0) \cite{shrouf2014smart}. 
Industrial IoT (IIoT) provides interconnection to smart factories by connecting different types of industrial machines and devices,
which helps to realize the intelligent manufacturing. 
To deal with the huge number of IIoT devices, a traditional centralized architecture is applied 
to provide services for IIoT devices, where IIoT devices are connected to a cloud server through the internet.
With the rapid growth of the number of IIoT devices and the performance requirement of the IIoT applications, 
however, the traditional centralized IIoT architecture faces many challenges, such as security, personal 
privacy, bandwidth constraint, and service delay \cite{rehman2019cloud}. 
To avoid these issues, some works introduce decentralized 
peer-to-peer (P2P) architectures for IIoT applications \cite{krco2005p2p,mietz2013p2p,chung2016p2p}, 
where each device can exchange information or trade
directly with other devices without a third-party organization. 
However, these P2P architectures still  face with 
security and personal privacy issues.

In the past few years, blockchain, as a world-changing technology, has shown its excellent properties in many fields
\cite{novo2018blockchain,xu2019blockchain-crowdsourcing,yang2019guest,wan2019blockchain}.
Blockchain is a decentralized public ledger that stores data in a list of blocks, these blocks are linked using 
cryptography, each block stores the hash value of the previous block. No centralized server is required for
maintaining the blockchain, instead, all the blocks are copied and shared by each user in the blockchain network,
and the blockchain is maintained by multiple participants in the P2P network. Blockchain provides a novel 
technology that helps establish a secure, trusted and decentralized system for solving the security and 
personal privacy problems in IIoT applications. To generate a new block, the participants (miners) of the blockchain 
need to solve a proof-of-work (PoW) puzzle \cite{nakamoto2019bitcoin}, which is hard to be solved but easy to be validated. 
The one who
first solve the puzzle has the right to packet a new block (the process is called mining), and will get a reward from
the platform. Nevertheless, the lightweight 
IIoT devices cannot directly participate in the mining process due to the huge computational resources requirement.

The edge computing architecture makes it possible for IIoT applications to establish a blockchain network, where IIoT devices
can offload the computational tasks to edge servers \cite{liu2018computation,guo2019blockchain,liu2018distributed}. 
Specifically,  incentivized by the reward from the platform for 
packeting a new block, each IIoT device will purchase a certain amount of computational resources (such as CPU and GPU) 
from edge servers
to participate in the mining process for maximizing its own profit. Some existing works \cite{xiong2018mobile,yao2019resource} 
study the pricing problem
between resource provider (edge server) and miners (IIoT devices), i.e., the resource provider offer a price to maximize its own profit, 
and miners decide their demand for computational resource to maximize their payoffs.  However, these works 
didn't pay attention to the safety of the blockchain network. 
Generally, the blockchain platform will dynamically adjust
the threshold value of the hash puzzle to stabilize the generating speed of the new blocks. That is, if the total computational
power of all miners is small, the platform will give a large hash threshold value, otherwise, it will give a small hash threshold
value.
The total computational power of all miners will directly affect the safety of the blockchain network. An attacker who 
wants to tamper with the context in a block of the blockchain needs to solve the hash puzzle faster than the current whole network.
Thus it's more harder for the attacker to change the context of blocks if all miners provide more computational power
in the mining process.

In this paper, therefore, we study the incentive mechanism of the blockchain platform to motivate IIoT devices to purchase 
more computational resources to participate in the mining process, thus that a more secure blockchain network can be 
established for IIoT.
The main contributions of this paper are listed as follows.

\begin{itemize}
	\item We design an incentive mechanism for the IIoT blockchain platform where the blockchain platform provides a certain
	reward to participating IIoT devices, to attract IIoT devices to purchase more computational power from edge servers 
	to participate in the mining process, thereby building a more secure blockchain network. 
	\item We analyze the relationship between the security of the blockchain network and the total computational power of
	the entire network, and give the probability that an attacker can successfully tamper with the blockchain.
	\item  We formulate the interaction between blockchain platform and miners as a two-stage Stackelberg game. We 
	analyze the existence and uniqueness of the Stackelberg equilibrium, and propose an efficient algorithm to compute 
	the Stackelberg equilibrium point.
	\item We conduct extensive simulations to evaluate the performance of our proposed algorithm, and we analyze
	the strategies of platform and miners in different situations. Our work is helpful for the IIoT blockchain
	platform to set a reasonable reward pricing strategy to maximize its utility, which is closely related to the security
	of the blockchain network.
\end{itemize}

The remainder of this paper is structured as follows. 
In Section \ref{Sec:related}, we introduce the related works of this paper. In Section \ref{sec:Problem}, we describe 
the system model and analyze the blockchain security that motivated our problem, and then we formulate our
problem as a two-stage Stackelberg game. In Section \ref{sec:Game analysis}, we analyze the existence and 
uniqueness of the Stackelberg equilibrium, and give the best strategies for miners and blockchain platform.
We conduct extensive performance evaluations in Section \ref{sec:Simulation}. And finally, we conclude this paper
in Section \ref{sec:Conclusions}.

\section{related works} \label{Sec:related}
Recently, there are numerous works concentrate on the IoT blockchain networks. 
Huang \emph{et al.} \cite{huang2019building} build a redactable consortium blockchain based on the first threshold 
chameleon hash and accountable and sanitizable chameleon signature schemes, which is efficient for lightweight IIoT
devices to operate.
Fu \emph{et al.} \cite{fu2019cooperative} propose an integrated architecture of cooperative computing
to support the demand of computing power in IoT blockchain networks, their goal is to maximize the system energy efficiency.
Xu \emph{et al.} \cite{xu2019become} propose a blockchain-enabled computation offloading method to ensure data integrity
for IoT in mobile edge computing.
To overcome the huge operational overhead in energy trading market in the industrial IoT, which is resulted by the 
frequent transactions, Hou \emph{et al.} \cite{hou2019local} investigate the local electricity storage for the blockchain-based 
energy trading in the IIoT. Their solutions can achieve a good tradeoff between credit utility and operational 
overhead. 
The authors in \cite{viriyasitavat2019new} integrate technologies such as service-oriented architecture, public
key infrastructure and enablers for Service Selection with blockchain technology, and propose a new 
blockchain-based architecture for service interoperation in IoT, which ensures data validity and guarantees the trust
of quality of service attributes for service selection.
Xu \emph{et al.} \cite{xu2019blockchain} propose a blockchain-based fair nonrepudiation network computing service 
provisioning scheme for IIoT, in which the blockchain is used as a service publication proxy and an evidence 
recorder. Their solutions overcome the drawbacks of traditional IIoT systems, where malicious services or 
clients may cheat others due to their own interests.

Moreover, there are a series of works study the blockchain from the aspect of auction or game theory.
Yao \emph{et al.} \cite{yao2019resource} model the resource management and pricing problem as a Stackelberg game,
and they design a multiagent reinforcement learning algorithm to search the near-optimal policy. 
Li \emph{et al.} \cite{li2017consortium} propose a secure energy trading system for industrial IoT, in which they use 
Stackelberg game theory design an optimal pricing strategy scheme for energy-coin loans to maximize the profits of 
credit banks.
Jiao \emph{et al.} \cite{jiao2019auction} propose an auction-based market model for the trading between the cloud computing 
service provider and miners. Their purpose is to efficiently allocate computing resources to maximize the social welfare.
Wang \emph{et al.} \cite{wang2017novel} propose a blockchain and double auction mechanism-based decentralized 
electricity transaction mode for microgrids, to achieve secure and quick electricity transactions.
Xiong \emph{et al.} \cite{xiong2018cloud} formulate the interaction between the cloud providers and miners as a Stackelberg
game, and apply backward induction to analyze the equilibria in each sub-game. 


\begin{figure} 
	\centering
	\includegraphics[width=.4\textwidth]{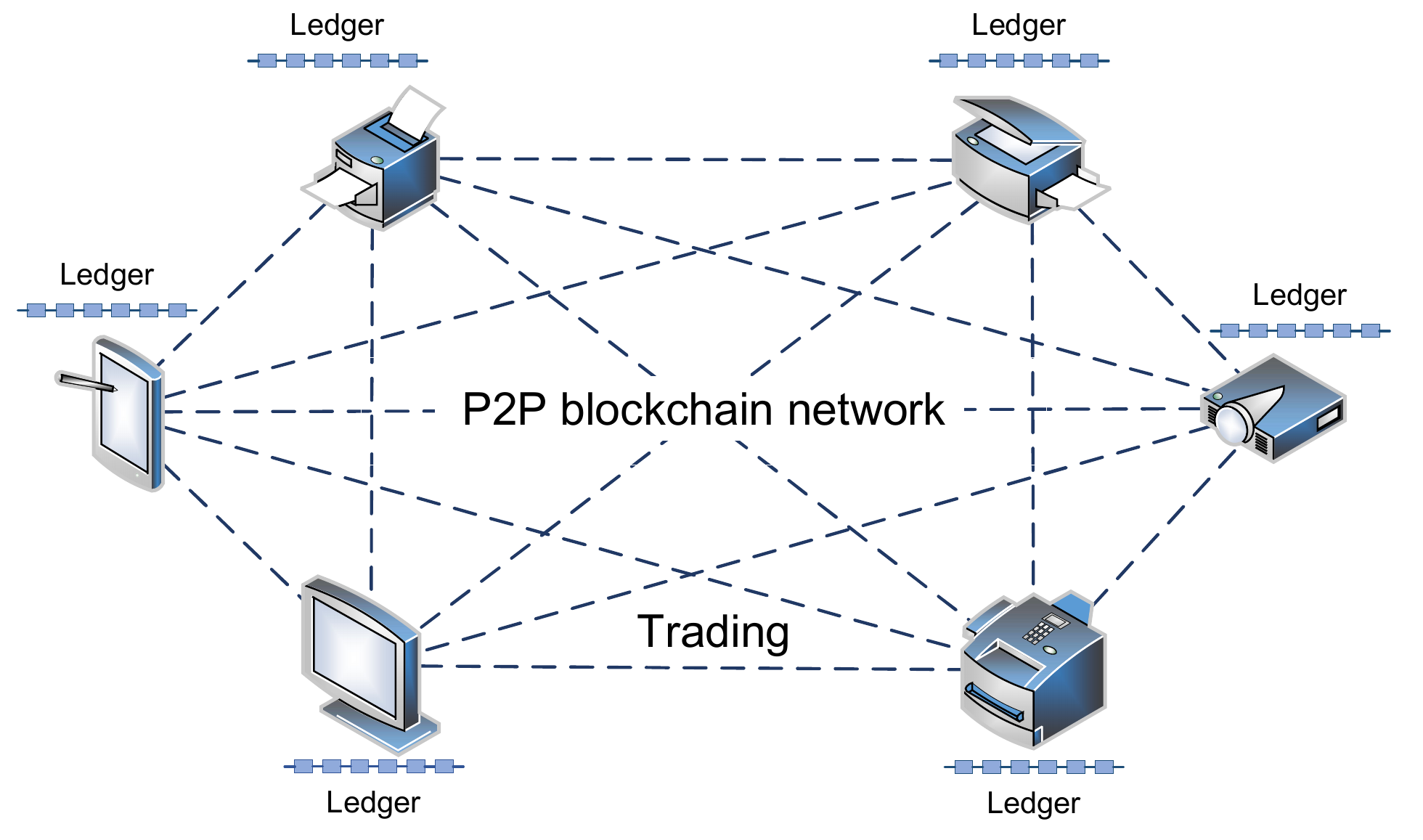}
	\caption{ The P2P blockchain network structure \cite{biswas2018scalable}.}
	\label{fig:p2pBlock}
\end{figure}

\begin{figure} 
	\centering
	\includegraphics[width=.4\textwidth]{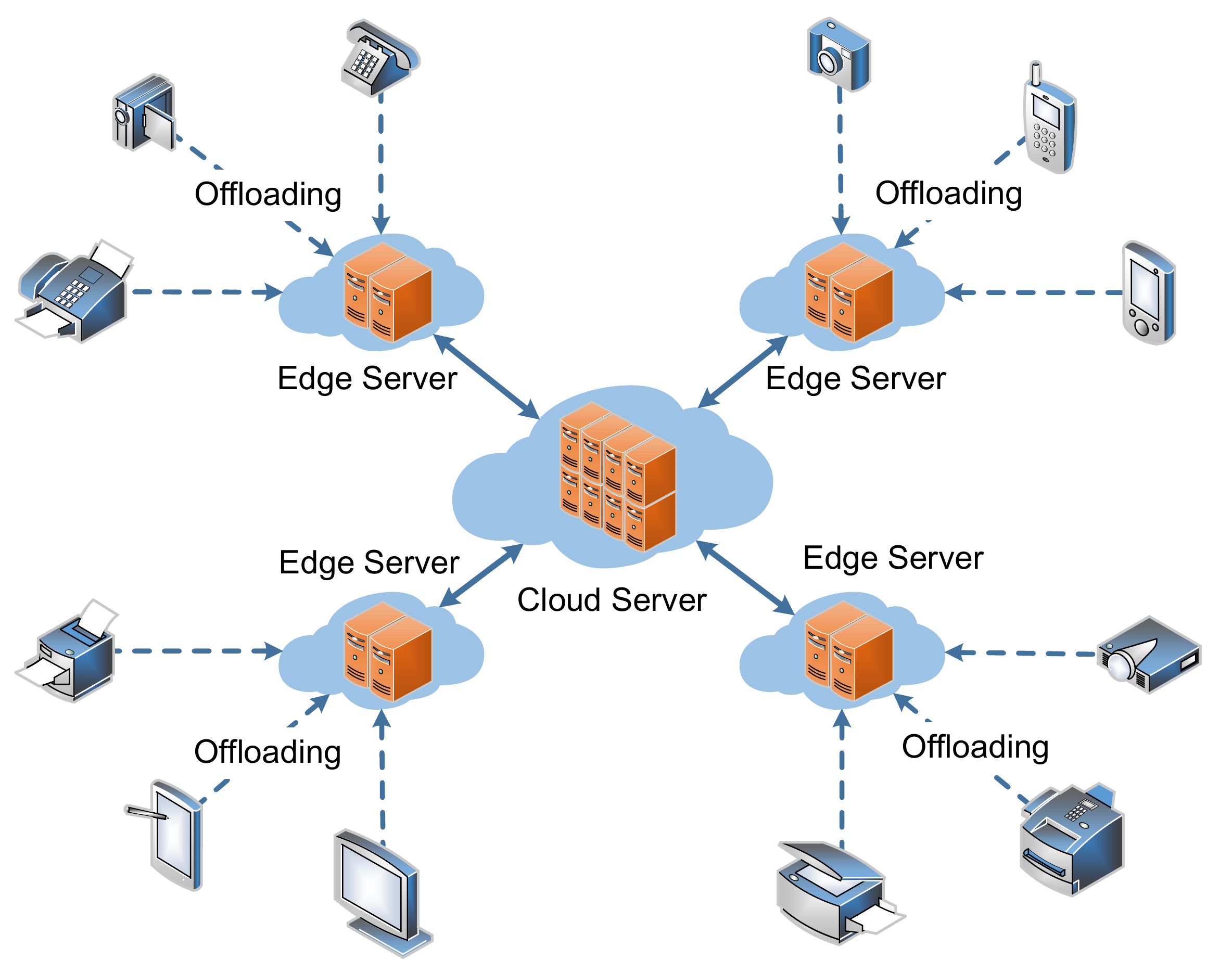}
	\caption{ The architecture of edge computing \cite{luo2020edge}.}
	\label{fig:EdgeComputing}
\end{figure}

\section{System model and problem formulation }\label{sec:Problem}
In this section, we first introduce the edge computing system of the IIoT about blockchain network,
we then analyze the security of the blockchain network, and finally, 
we formulate the incentive problem between blockchain platform and miners.
\subsection{System Model} \label{sec:Model}
In the blockchain network, the core problem is to achieve distributed consensus. Satoshi Nakamoto proposed the 
PoW consensus protocol in 2008 which is used for Bitcoin \cite{nakamoto2019bitcoin}. In the PoW consensus,
the users who want generate a new block need to solve a hash puzzle, which is very costly to be addressed but
easy for others to verify. This process is called mining, and these users are termed as miners. These miners compete
with each other to solve a hash puzzle, the one who first solve the hash puzzle has the right to generate a new 
block and will get a reward from the blockchain platform.

For the IIoT blockchain network, the lightweight IIoT devices cannot directly participate the mining process due to the 
limited computational capability. Incentivized by the reward from the blockchain platform, IIoT devices will purchase
computation resource from edge servers, each edge server offers its own unit price for computational resource. 
As shown in Fig. \ref{fig:p2pBlock}, the blockchain network is maintained by all of the IIoT devices. As for the mining process, 
each miner will offload its computational task to the edge server to compete with others, as shown in Fig. 
\ref{fig:EdgeComputing}. The probability of each miner
winning the competition depends on the amount of computational resource it purchased. All the miners purchase 
computational resource with the goal of maximizing their own profits. 

\subsection{Blockchain Security} \label{blockchain security}
Blockchain is a list of blocks that are linked by block hash value, each block record a set of transactions. 
More specifically, a block contains two parts: block content and block header. 
The block content is the details of transactions information, which records all the inputs and outputs of each transaction. 
The block header consists of the previous block hash value, which is used as a cryptographic link that creates the chain,
a version number that used for tracking for software or protocol updates, 
a timestamp that records the time at which the block is generated, a Merkle tree root of all the transactions, a
hash threshold value that records the current mining difficulty, and a nonce, which is used for solving the PoW puzzle.

The blockchain starts with a genesis block which is given by the blockchain platform, all subsequent blocks will put 
some previously generated block's hash value into their block header.
Miners compete with each other to 
solve a hash puzzle, the one who first solve the puzzle has the right to generate a new block, and new blocks will be added 
behind the genesis block. Forks may happen when multiple miners solve the hash puzzle at the same time, thus each user
maintains the blocks in the form of a block tree \cite{hari2019accel}. 
According to the longest chain principle 
\cite{nakamoto2019bitcoin}, each user will choose the longest branch in the block tree as the current blockchain, as shown
in the left part of Fig. \ref{fig:Attacker-blockchain}.

\begin{figure}
	\centering
	\includegraphics[width=.48\textwidth]{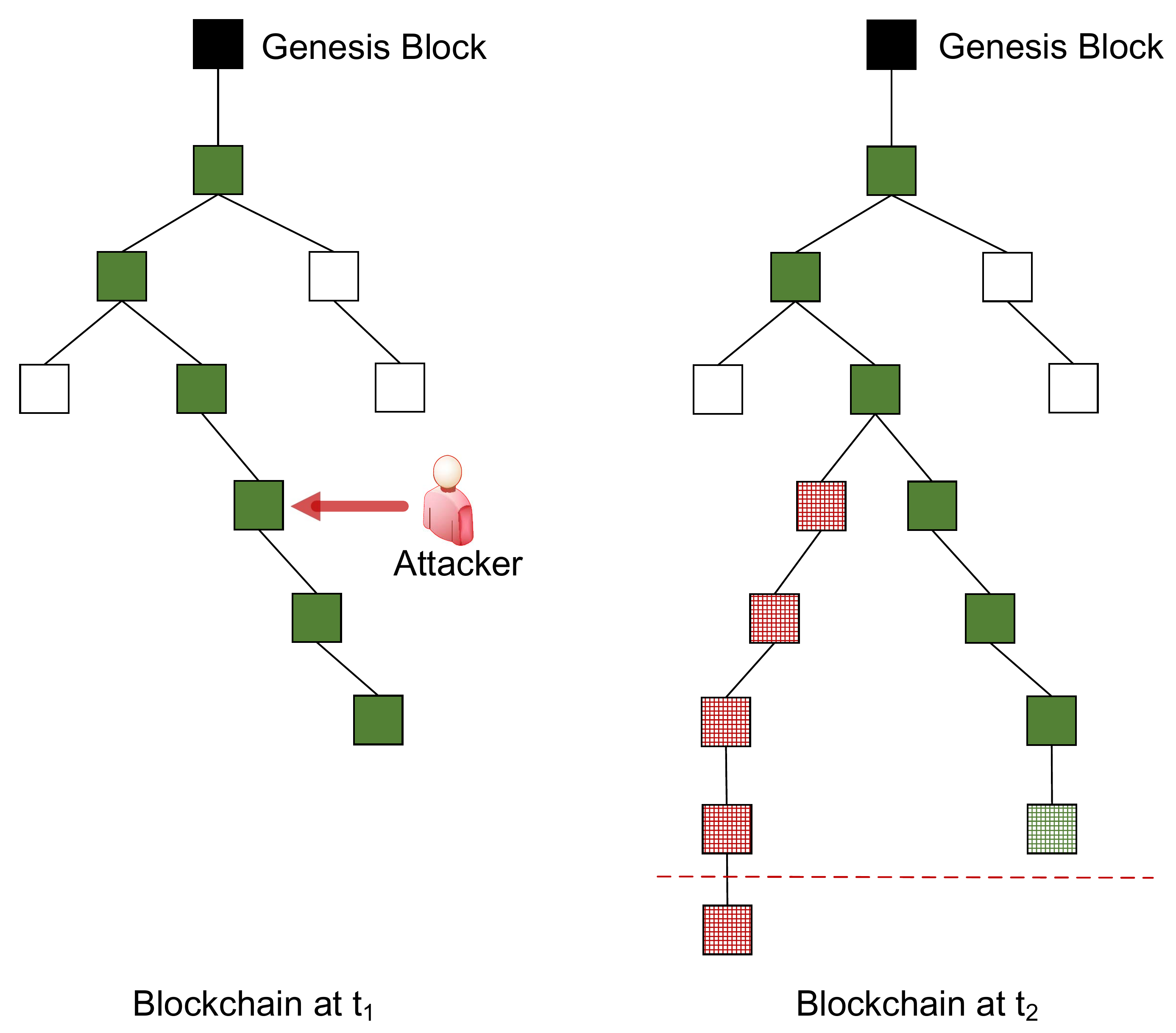}
	\caption{ Attackers tamper with the blockchain by winning the block mining race.}
	\label{fig:Attacker-blockchain}
\end{figure}

For a blockchain miner, to get the right to attach a new block to the chain, it needs to solve a hash function (PoW puzzle), that is, 
it needs to find a nonce and record it in the block header such that the hash value of the block header is less than 
the hash threshold.
As the hash function has no back door, the only way to find such a nonce is to run many hash operations. 
The difficulty of the PoW puzzle is determined by the given hash threshold value, which is a 256-bit binary number that starts 
with a certain number of consecutive zeros (difficulty). 
For example, if the hash threshold starts with 60 consecutive zeros, the probability of finding the correct nonce by 
performing a hash operation is $2^{-60}$, which means that it takes an average of $2^{60}$ hash operations to solve 
the PoW puzzle. To stabilize the growth rate of the blockchain, the platform will dynamically adjust the difficulty of the PoW puzzle  
based on the total computational power of the whole blockchain network. Take the bitcoin blockchain as an example, the
difficulty of PoW puzzle will be updated every 2 weeks to ensure that it takes 10 minutes (on average) to add a new 
block to the blockchain \cite{bitcoinDifficuty}.

Assume that an attacker wants to tamper with the context of a block, the change of the context will change the hash 
value of the block, so the attacker needs to find a new nonce to solve the PoW puzzle of this block. Moreover,
as each block contains the hash value of the previous block, the change of a block context will change all the subsequent 
blocks, so the attacker should find the nonce for every subsequent block.
In fact, the attacker needs to fork a new branch, and start a block mining race against other miners of the network.
The attacker successfully tamper with the blockchain once the attacker wins the race, as the new branch forked by the attacker
becomes the longest 
one in the block tree. As shown in Fig. \ref{fig:Attacker-blockchain}.

Now we analyze the probability of success of an attacker. Assume that the total computational power of the blockchain network
is $H$ by the hash rate, and the attacker's computation power is $h$. Then the probability an honest miner finds the next block 
can be denoted by $p = \frac{H-h}{H}$, and the probability that the attacker finds the next block is $q = \frac{h}{H}$.
According to \cite{grunspan2018double}, the probability that the attacker will ever win the race from $z$ blocks behind can
be calculated by

\begin{gather}
	P(z) = 
	\begin{cases}
		1, &\text{if } q \geq \frac{1}{2}, \\
		I_{4pq}\left(z, \frac{1}{2}\right), &\text{otherwise,}
	\end{cases}
\end{gather}
where $I_{x}(u,v)$ is the regularized incomplete beta function:
\begin{gather}
	I_{x}(u,v) = \frac{\Gamma(u+v)}{\Gamma(u)\Gamma(v)} \int_{0}^{x} t^{u-1} (1-t)^{v-1} \, dt,
\end{gather}
and $\Gamma(\cdot)$ is the gamma function. 

\begin{figure}
	\centering
	\includegraphics[width=.4\textwidth]{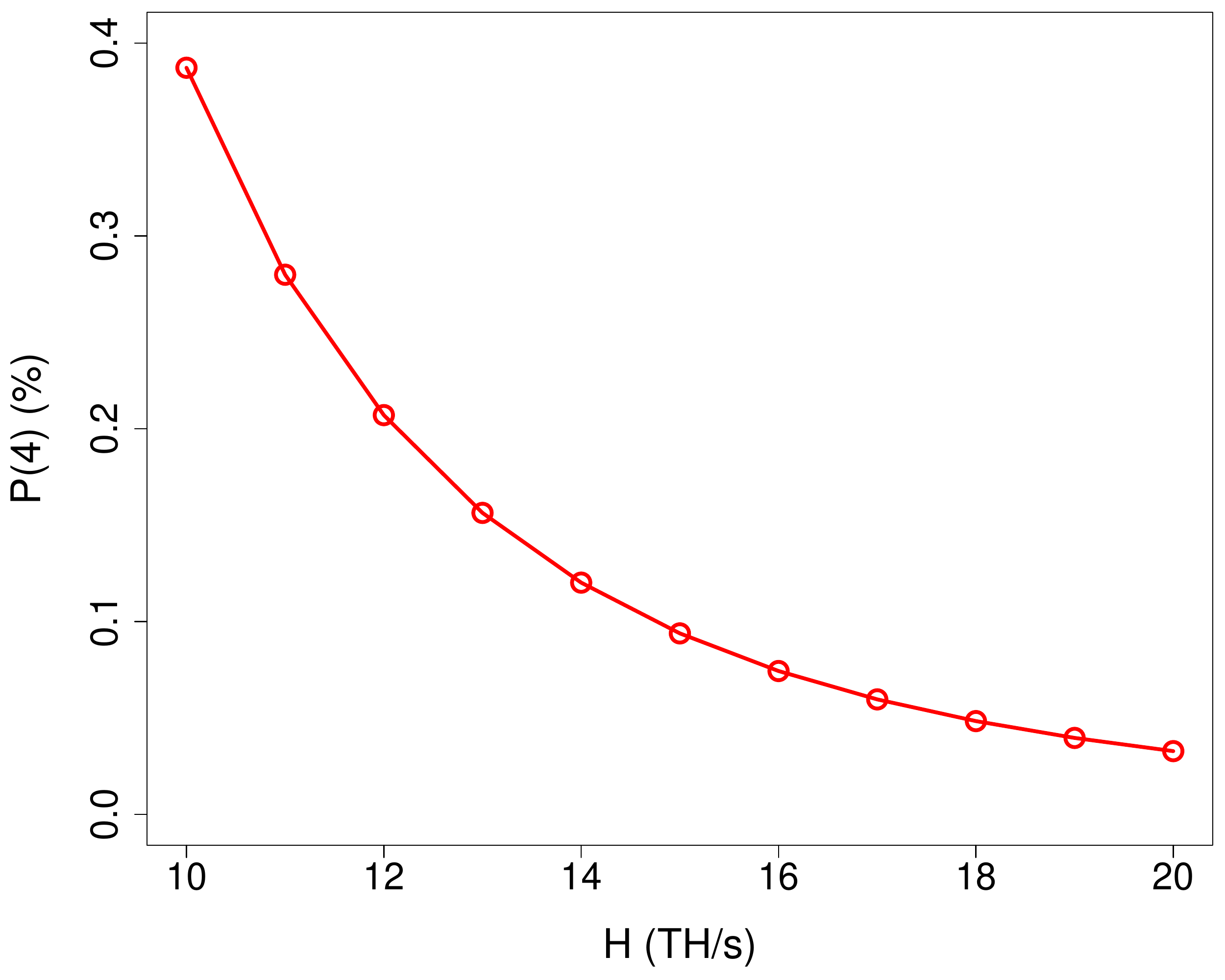}
	\caption{ The probability that an attacker will ever win the race from 4 blocks behind.}
	\label{fig:attackerPro}
\end{figure}

Consider that an attacker's computational power is $h = 1 ~TH/s$, and he wants to tamper with the context of the $4th$ block 
from the blockchain tip. Fig. \ref{fig:attackerPro} shows the probability of success of the attacker against the total computational
power of the network. It can be seen that as the total computational power of the network increases, the probability that 
the attacker successfully tamper with the blockchain significantly decreases.
Specifically, when $H = 10~TH/s$, $P(4) = 0.387\%$; when $H = 20~TH/s$, $P(4) = 0.033\%$.
The probability of success of the attacker reduced by more than 10 times if we double the computational power of the entire
network.
The result suggests that the larger the computational power of the whole blockchain network, 
the harder it is for the attacker to tamper with the blockchain, and thus the more secure the blockchain network will be.

\subsection{Problem Description}
As described before, the blockchain platform provides certain rewards to incentivize miners to participate in the mining process,
and each miner purchase computational resource from its nearby edge server for getting more profits. 

For the blockchain platform, the revenue is mainly derived from the transaction fees, assume that the average transaction
fees in a block is $B$. The platform will give a reward $R$ to the miner who packaging a new block. Note that the reward
$R$ should be no larger than $B$, otherwise the blockchain cannot maintain perpetual operation. As discussed in Section
\ref{blockchain security}, the total computational power will significantly affect the blockchain security, so the blockchain
platform will get benefit from the huge computational power provided by the miners.
As shown in Fig .\ref{fig:attackerPro}, the decrease speed of the probability that an attacker successfully tamper with the 
blockchain slows down as the 
total computational power of the entire network increases. Thus we use the \emph{sigmoid function} to describe the benefits of the 
total computational power of the entire network.
We define the utility of the blockchain platform as 
\begin{gather} \label{equ:utility1}
	U = \alpha \cdot \left [\sigma \left(\beta \cdot \sum\nolimits_{s_i\in S} \mu_i\right) - \frac{1}{2} \right]  - R,
\end{gather}
where $S$ is the miners set, $\mu_i$ is the computational power provided by the $i$-$th$ miner, 
$\sigma(\cdot)$ is the \emph{sigmoid function} that defined as $\sigma(x) = \frac{1}{1 + e^{-x}}$,
$\alpha > 1$ is a constant parameter that controls the importance of the total computational power,
and $0 < \beta < 1$ is a constant parameter that controls the convergence rate of the sigmoid function.

For the miners, consider there are a set $S$ of IIoT devices that are interested in participating in the mining process, where
$S = \{s_1, s_2, \dots, s_n\}$. Each miner $s_i \in S$ will purchase $\mu_i$ computational power from edge servers to 
compete with others in pursuit of a maximum financial profit. 
According to the principle of the PoW consensus protocol, the speed of a miner
for solving the hash puzzle depends on the number of hash operations it can perform per unit of time, that is, its computational
power. Therefore, the miner who purchases the most computational power is most likely to solve the PoW puzzle first.
We use $p_i$ to denote the probability that miner $s_i$ winning the competition among all miners in solving the PoW puzzle,
$p_i$ can be defined as
\begin{gather}
	p_i = \frac{\mu_i}{\sum\nolimits_{s_j \in S} \mu_j}.
\end{gather}

The unit price of the computational power purchased by miners may be different, as they purchase computational power
from different edge servers.
Assume that the unit price of the computational power purchased by $s_i$ is $\lambda_i$ per day.
We also assume that the blockchain will generate an average of $N$ new blocks per day.
Then the expected reward of miner $s_i$ in a day is $p_i R N$, and its cost is $\lambda_i \mu_i$.
We use $P_i$ to represent the expected profit of miner $s_i$ in one day, $P_i$ can be calculated as follows,
\begin{gather} \label{equ: profit_miner}
	P_i = p_i R N - \lambda_i \mu_i. 
\end{gather}

Both blockchain platform and miners will dynamically adjust their strategies to get the maximum profit.
We model the interaction between blockchain platform and miners as a two-stage Stackelberg game. In the upper stage, the
blockchain platform sets the reward to incentivize miners participate the mining process. In the lower stage, miners decide
the optimal amount of computational power they purchase. We formulate the optimization problems for the blockchain
platform and miners as follows.

We first introduce the lower stage of the game. Given the reward $R$ of the blockchain platform and other miners' strategies
$\bm{\mu_{-i}}$, where $\bm{\mu_{-i}} = \{\mu_1, \mu_2, \dots, \mu_{i-1}, \mu_{i+1}, \dots, \mu_n\}$.
The miner $s_i$ decides the amount of computational power $\mu_i$ it purchased to maximize its own profit.
This sub-game problem can be written as follows.
\begin{problem}
	miners' sub-game.
	\begin{align*}
   		 &\max\limits_{\mu_i} ~~~ P_i(\mu_i | \bm{\mu_{-i}}, R)\\
    		&~ s.t.  ~~~ \mu_i \geq 0
	\end{align*}
\end{problem}

For the upper stage of the game, the blockchain platform will dynamically adjust the reward $R$ to maximize its utility.
As defined in equation (\ref{equ:utility1}), the utility of the blockchain platform is directly related to the strategies $\bm{\mu}$
of miners, where $\bm{\mu} = \{\mu_1, \mu_2, \dots, \mu_n\}$, 
and the reward $R$. 
This sub-game can be formulated as follows.
\begin{problem}
	blockchain platform's sub-game.
	\begin{align*}
   		 &\max\limits_{R} ~~~ U(R, \bm{\mu})\\
    		& ~s.t. ~~~ 0 \leq R \leq B
	\end{align*}
\end{problem}

Problem 1 and Problem 2 together form a Stackelberg game. The object of the game is to find a \emph{Stackelberg equilibrium} 
point where neither the leader (blockchain platform) nor the followers (miners) want to change their strategies.
In this paper, the Stackelberg equilibrium can be defined as follows.
\begin{define}
    Let $\bm{\mu^*}$ and $R^*$ be the optimal strategies of miners and blockchain platform, respectively, where
    $\bm{\mu^*} = \{\mu_1^*, \mu_2^*, \dots, \mu_n^*\}$. 
    Then, the point ($\bm{\mu^*}$, $R^*$) is the Stackelberg equilibrium point if it satisfies the following two conditions,
     \begin{gather}
    	U(R^*, \bm{\mu^*}) \geq U(R, \bm{\mu^*}), \forall ~0 \leq R \leq B,
    \end{gather}
    and
    \begin{gather}
    	P_i(\mu_i^* | \bm{\mu_{-i}^*}, R^*) \geq P_i(\mu_i | \bm{\mu_{-i}^*}, R^*), \forall i, \forall \mu_i \geq 0,
    \end{gather}
    where $\bm{\mu_{-i}^*} = \bm{\mu^*} \backslash \{\mu_i^*\}$.
\end{define}


\section{Game analysis for the incentive mechanism}\label{sec:Game analysis}
In this section, we analyze the existence and the uniqueness of the Stackelberg equilibrium of our proposed Stackelberg
game. We first analyze the lower stage of the game, in which we aim to find the \emph{Nash equilibrium} for the miners'
sub-game. Based on the analysis of the lower stage, we then analyze the utility maximization of the blockchain platform’s 
sub-game in the upper stage.

\subsection{Analysis of the miners's sub-game} \label{subsec:miners}
After the blockchain platform set the reward $R$ for miners, all of the miner will dynamically adjust their strategies to get the
maximum profits until reach a Nash equilibrium. In the following, we will prove that the Nash equilibrium point exists in the 
miners' sub-game through a theorem.
\begin{theorem}
	The Nash equilibrium point exists in the miners' sub-game.
\end{theorem}
\begin{proof}
	For the miners' sub-game, the object function $P_i(\cdot)$ is defined in $[0, \infty)$. From equation (\ref{equ: profit_miner}),
	we can know that $\mu_i \leq \frac{RN}{\lambda_i}$, otherwise, the profit of miner $s_i$ will be negative. Thus $\mu_i$ 
	is continuously chosen in $[0, \frac{RN}{\lambda_i}]$, which is a non-empty, convex and compact subset of the Euclidean
	space. 
	Next, we calculate the first order and second order derivatives of function $P_i(\cdot)$.
	\begin{gather}
		\frac{\partial P_i}{\partial \mu_i} =  RN\frac{\sum\nolimits_{s_j \in S} \mu_j - \mu_i} {(\sum\nolimits_{s_j\in S} \mu_j)^2}
		- \lambda_i,
	\end{gather}
	\begin{gather} \label{equ7:concave}
		\frac{\partial^2 P_i}{\partial \mu_i ^2} = \frac{\partial (\frac{\partial P_i}{\partial \mu_i})} {\partial \mu_i}
		= -2RN \frac{\sum\nolimits_{s_j \in S} \mu_j - \mu_i}{(\sum\nolimits_{s_j\in S} \mu_j)^3} \leq 0.
	\end{gather}
	
	Therefore, $P_i(\cdot)$ is a strictly concave function, and we then conclude that the Nash equilibrium point exists in the 
	miners' sub-game.
\end{proof}

As $P_i(\cdot)$ is a strictly concave function with $\mu_i$, given the reward $R$ of the blockchain platform and other miners' 
strategies $\bm{\mu_{-i}}$, miner $s_i$ has a unique best strategy $u_i$, and it can be achieved when the first order
derivative of $P_i(\cdot)$ equals to 0, i.e.,
\begin{gather} \label{equ: firstorder-u}
	\frac{\partial P_i}{\partial \mu_i} =  RN\frac{\sum\nolimits_{s_j \in S} \mu_j - \mu_i} {(\sum\nolimits_{s_j\in S} \mu_j)^2}
		- \lambda_i = 0.
\end{gather}

Solving equation (\ref{equ: firstorder-u}), we have 
$
\mu_i = \sqrt{\frac{RN\sum\nolimits_{s_j \in S\backslash \{s_i\}}\mu_j}{\lambda_i}} - \sum\limits_{s_j \in S\backslash \{s_i\}}\mu_j.
$
As each strategy $\mu_i$ for $s_i$ is an nonnegative number, the best strategy $\mu_i^*$ for $s_i$ is given as follows.
\begin{gather} \label{equ:beststrategy}
	\mu_i^* = 
	\begin{cases}
    		\displaystyle 0, ~~~~~~~~~~~~~~~~~~~~\text{if }  \frac{RN} {\sum\nolimits_{s_j\in S\backslash \{s_i\}} \mu_j} \leq \lambda_i,\\
		\sqrt{\frac{RN \cdot \sum\limits_{s_j \in S\backslash \{s_i\}}\mu_j}{\lambda_i}} - \sum\limits_{s_j \in S\backslash \{s_i\}}\mu_j, \text{otherwise}
	\end{cases}
\end{gather}

\begin{corollary} \label{corollary1}
    	Given the optimal strategies of miners $\bm{\mu^*} = \{\mu_1^*, \mu_2^*, \dots, \mu_n^*\}$, for any $\mu_i^*, \mu_j^*
	\in \bm{\mu^*}$, if $\lambda_i \leq \lambda_j$, then $\mu_i^* \geq \mu_j^*$.
\end{corollary}
\begin{proof}
	We prove this corollary by contradiction. Suppose that there exist two miners' strategies
	 $\mu_i^*, \mu_j^* \in \bm{\mu^*}$, where 	$\lambda_i \leq \lambda_j$ and $\mu_j^* > \mu_i^* \geq 0$.
	 As $\mu_j^* > 0$, according to equations (\ref{equ: firstorder-u}) and (\ref{equ:beststrategy}),  we can easily know that
	 $\frac{\partial P_j}{\partial \mu_j}(\mu_j^*|\bm{\mu_{-j}^*}) = 0$.
	 
	 Similarly, If $\mu_i^* > 0$, we have $\frac{\partial P_i}{\partial \mu_i}(\mu_i^*|\bm{\mu_{-i}^*}) = 0$.
	 If $\mu_i^* = 0$, according to equation (\ref{equ:beststrategy}), we have 
	 $\frac{RN} {\sum\nolimits_{s_j\in S\backslash \{s_i\}} \mu_j} \leq \lambda_i$, substituting it in to equation 
	 (\ref{equ: firstorder-u}), we have $\frac{\partial P_i}{\partial \mu_i}(\mu_i^*|\bm{\mu_{-i}^*}) \leq 0$.
	 In summary, we know that $\frac{\partial P_i}{\partial \mu_i}(\mu_i^*|\bm{\mu_{-i}^*}) \leq 0$.
	 
	 As $\lambda_i \leq \lambda_j$ and $\mu_j^* > \mu_i^* \geq 0$, we have
	  \begin{gather}
	  \begin{split}
	 	\frac{\partial P_j}{\partial \mu_j}(\mu_j^*|\bm{\mu_{-j}^*}) &= 
		RN\frac{\sum\nolimits_{s_k \in S} \mu_k^* - \mu_j^*} {(\sum\nolimits_{s_k\in S} \mu_k^*)^2} - \lambda_j 
		\\
	       &< RN\frac{\sum\nolimits_{s_k \in S} \mu_k^* - \mu_i^*} {(\sum\nolimits_{s_k\in S} \mu_k^*)^2} - \lambda_i \\
		&= \frac{\partial P_i}{\partial \mu_i}(\mu_i^*|\bm{\mu_{-i}^*}) \leq 0,
	\end{split}
	 \end{gather}
	 which contradicts $\frac{\partial P_j}{\partial \mu_j}(\mu_j^*|\bm{\mu_{-i}^*}) = 0$.
	 Therefore, the corollary holds.
\end{proof}

Sort the miners in ascending order of $\lambda_i$, for clarity, miners are renumbered and still denoted by 
$S = \{s_1, s_2, \cdots, s_n\}$.
According to Corollary \ref{corollary1}, we have $\mu_1 \geq \mu_2 \geq \cdots \geq \mu_n \geq 0$.

We assume that the the first $q$ miners have a non-zero strategy, i.e., $\mu_q > 0$, $\mu_{q+1} = 0$.
We let $S_q = \{s_1, s_2, \dots, s_q\}$. It's obvious that $\sum\nolimits_{s_j \in S} \mu_j = \sum\nolimits_{s_j \in S_q} \mu_j$,
and we have the following corollary.

\begin{corollary} \label{corollary2}
	Let $q$ be the number of miners that have a non-zero strategy in a Nash equilibrium, then we have $q \geq 2$.
\end{corollary}
\begin{proof}
	Firstly, it's clear that $q > 0$, otherwise, according to equation (\ref{equ: profit_miner}), any miner $s_i$ who purchase 
	computational power with the amount in $(0, \frac{RN}{\lambda_i})$ can get more profit.
	Then we consider the case that $q = 1$. Let $s_1$ be the miner who has a positive strategy, and its current best 
	strategy is to purchase $\mu_1$ amount of computational power. According to equation (\ref{equ: profit_miner}), 
	$s_1$'s current profit is $RN - \lambda_1 \mu_1$.
	However, $s_1$ can increase its profit by continuously reducing $\mu_1$ to 0, which indicates that miners didn't reach a 
	Nash equilibrium point.
	Therefore, the corollary holds.
\end{proof}

Summing up equation (\ref{equ: firstorder-u}) with $i = 1, 2, \dots, q$, we have
\begin{gather} \label{equ11}
	\frac{RN(q-1)}{\sum\nolimits_{s_j \in S_q} \mu_j} - \sum\nolimits_{s_i \in S_q}\lambda_i = 0.
\end{gather}
Thus we have
\begin{gather} \label{equ:sum-u}
	\sum\nolimits_{s_j \in S_q} \mu_j = \frac{RN(q-1)}{\sum\nolimits_{s_i \in S_q}\lambda_i}.
\end{gather}
Substituting equation (\ref{equ:sum-u}) into equation (\ref{equ: firstorder-u}), we obtain
\begin{gather}\label{equ:beststra-u}
	\mu_i = \frac{RN(q-1)}{\sum\nolimits_{s_j \in S_q}\lambda_j} \left ( 1 - \frac{(q-1)\lambda_i}{\sum\nolimits_{s_j \in S_q}\lambda_j} \right).
\end{gather}

As $\mu_q > 0$, and we have proven that $q \geq 2$, we then have $ 1 - \frac{(q-1)\lambda_q}{\sum\nolimits_{s_j \in S_q}
\lambda_j} > 0$, i.e., $\lambda_q < \frac{\sum\nolimits_{s_j \in S_{q-1}}\lambda_j}{q-2}$. 

Based on the above analysis, we design the following algorithm to find the Nash equilibrium point for the miners' sub-game.

\begin{algorithm} 
	\caption{Calculate Nash equilibrium for miners.}
	\label{algorithm:miner}
	\begin{algorithmic}[1]
		\State Sort miners in ascending order of $\lambda_i$ and renumber miners, i.e., 
		$\lambda_1 \leq \lambda_2 \leq \dots \leq \lambda_n$.
		\State $S' = \{s_1, s_2\}$, q = 2
		\While {$q < n$ and $\lambda_{q+1} < \frac{\sum\nolimits_{s_i \in S'}\lambda_i}{|S'| - 1}$}
			\State $S' \leftarrow S' \cup \{s_{q+1}\}$ ;
			\State $q \leftarrow q+1$;
		\EndWhile
		\For{$i=1$; $i\leq n$; $i++$}
			\If{$s_i \in S'$}
				\State $\mu_i^* = \frac{RN(q-1)}{\sum\nolimits_{s_j \in S'}\lambda_j} \left ( 1 - \frac{(q-1)\lambda_i}{\sum\nolimits_{s_j \in S'}\lambda_j} \right)$;
			\Else 
				\State $\mu_i^* = 0$;
			\EndIf
		\EndFor \\
		\Return $\bm{\mu^*} = \{\mu_1^*, \mu_2^*, \dots, \mu_n^*\}$
	\end{algorithmic}
\end{algorithm}

In the following, we first prove the strategies produced by Algorithm \ref{algorithm:miner} is a Nash equilibrium for the miners'
sub-game, then we prove that the Nash equilibrium is unique.

\begin{theorem}
	The  strategies produced by Algorithm \ref{algorithm:miner} is a Nash equilibrium for the miners' sub-game.
\end{theorem}
\begin{proof}
	For miners in $S'$, their strategies are calculated by equation (\ref{equ:beststra-u}), it's clear that these miners 
	get the current best strategies as the first order of $P_i(\cdot)$ equals to $0$ for $s_i \in S'$.
	To prove the theorem, we only need to prove that for any miner $s_j \in S \backslash S'$, its current best strategy is $0$.
	According to the description of Algorithm \ref{algorithm:miner}, we have 
	\begin{gather} \label{eql:14}
		\lambda_j \geq \frac{\sum\nolimits_{s_i \in S'}\lambda_i}{|S'| - 1},  \forall s_j \in S \backslash S'.
	\end{gather}
	From equation (\ref{equ:sum-u}), we know that
	$\sum\nolimits_{s_i \in S'}\lambda_i = \frac{RN(|S'|-1)}{\sum\nolimits_{s_i \in S'}\mu_i }$, substituting it into
	equation (\ref{eql:14}), we obtain that
	$\frac{RN}{\sum\nolimits_{s_i \in S'}\mu_i } \leq \lambda_j$ for any $s_j \in S \backslash S'$.
	As $s_j \notin S'$, we have $\sum\nolimits_{s_i \in S'}\mu_i = \sum\nolimits_{s_i \in S\backslash \{s_j\}}\mu_i $.
	Therefore, $\frac{RN}{\sum\nolimits_{s_i \in S\backslash \{s_j\}}\mu_i } \leq \lambda_j$ for any $s_j \in S \backslash S'$.
	According to equation (\ref{equ:beststrategy}), we know that the current best strategy for any miner 
	$s_j \in S \backslash S'$ is $0$. Thus the theorem holds.
\end{proof}

The following corollary helps us to prove the uniqueness of the Nash equilibrium for the miners' sub-game.
\begin{corollary} \label{corollary3}
	Given any Nash equilibrium $\bm{\mu^{ne}}=\{\mu^{ne}_1, 
	\mu^{ne}_2, \dots, \mu^{ne}_n\}$ for the miners' sub-game, let $S_h$ be the set of miners with a non-zero 
	strategy, then we have $S_h = S'$, where $S'$ is got by Algorithm \ref{algorithm:miner}, and $|S'| = p$.
\end{corollary}
\begin{proof}
	Assume that miners have been sorted in ascending order of $\lambda_i$. According to
	Corollary $\ref{corollary1}$, we know that $\mu^{ne}_1 \geq \mu^{ne}_2 \geq \dots \geq \mu^{ne}_n \geq 0$.
	Suppose that $S_h = \{s_1, s_2, \dots, s_h\}$, i.e. $|S_h| = h$. To prove $S_h = S'$, we only need to prove that
	$h = p$. If $h > p$, then $s_{p+1} \in S_h$, from description of Algorithm \ref{algorithm:miner}, we have 
	$\lambda_{p+1} \geq \frac{\sum\nolimits_{s_i \in S'}\lambda_i}{|S'| - 1} \geq \frac{\sum\nolimits_{s_i \in S_h}\lambda_i}
	{|S_h| - 1}$, substituting this into equation (\ref{equ:beststra-u}), we obtain that $\mu_{p+1}^{ne} \leq 0$, which contradicts
	$s_{p+1} \in S_h$. If $h < p$, then we have $\mu_{h+1}^{ne} = 0$ and $\lambda_{h+1} < \frac{\sum\nolimits_{s_i \in S_h}
	\lambda_i}{|S_h|-1}$, according to equation (\ref{equ: firstorder-u}), 
	the first order derivative of $P_{h+1}(\cdot)$ with respect to $\mu_{h+1}$ when $\mu_{h+1} = \mu_{h+1}^{ne} = 0$ is
	\begin{gather} \label{equ:15}
		\begin{split}
			&\frac{\partial P_{h+1}}{\partial \mu_{h+1}}(0|\bm{\mu^{ne}_{-(h+1)}}) 
			= RN\frac{\sum\nolimits_{s_j \in S} \mu_j - 0} {(\sum\nolimits_{s_j\in S} \mu_j^{ne})^2}
			- \lambda_{h+1} \\
			&= \frac{RN}{\sum\nolimits_{s_j\in S} \mu_j^{ne}} - \lambda_{h+1} = 
			\frac{RN}{\sum\nolimits_{s_j\in S_h} \mu_j^{ne}} - \lambda_{h+1},
		\end{split}
	\end{gather}
	where $\bm{\mu^{ne}_{-(h+1)}} = \bm{\mu^{ne}} \backslash \{\mu^{ne}_{h+1}\}$.
	
	According to equation (\ref{equ:sum-u}), we have $\sum\nolimits_{s_j \in S_h} \mu_j^{ne} = \frac{RN(h-1)}
	{\sum\nolimits_{s_j \in S_h}\lambda_j}$, substituting it into equation (\ref{equ:15}), we have that
	\begin{gather}
		\frac{\partial P_{h+1}}{\partial \mu_{h+1}}(0|\bm{\mu^{ne}_{-(h+1)}}) 
		= \frac{\sum\nolimits_{s_j \in S_h}\lambda_j}{h-1} - \lambda_{h+1}.
	\end{gather}
	As $\lambda_{h+1} < \frac{\sum\nolimits_{s_i \in S_h}\lambda_i}{|S_h|-1} = \frac{\sum\nolimits_{s_i \in S_h}\lambda_i}
	{h-1}$, we know that $\frac{\partial P_{h+1}}{\partial \mu_{h+1}}(0|\bm{\mu^{ne}_{-(h+1)}})  > 0$, which implies that 
	miner $s_{h+1}$ can improve its profit by increasing its strategy $\mu^{ne}_{h+1}$. This contradicts that $\bm{\mu^{ne}}$
	is an Nash equilibrium. Therefore, we have $h = p$, and thus the corollary holds.
\end{proof}

\begin{theorem} \label{theorem3}
	The miners' sub-game has a unique Nash equilibrium point.
\end{theorem}
\begin{proof}
	According to Corollary \ref{corollary3}, 
	the miners' sub-game can be seen as a game among miners in $S'$, as for any miner $s_i \in S \backslash S'$
	we always have $\mu_i = 0$ in a Nash equilibrium. 
	Therefore, we only need to prove that the sub-game of miners in $S'$ has a unique Nash equilibrium point.
	
	As the strategy of each miner in $S'$ is positive, and the profit $P_i(\cdot)$ for any $s_i \in S'$ is a concave function
	according to equation (\ref{equ7:concave}), each miner will get its best strategy when the first order derivate of $P_i(\cdot)$
	equals to 0. As shown in equations (\ref{equ11})-(\ref{equ:beststra-u}), we get unique solutions by solving the set of
	functions that the first order derivate of $P_i(\cdot)$ equals to 0 for each $s_i \in S'$. Therefore, the miners' sub-game
	among miners in $S'$ has a unique Nash equilibrium, and thus we can conclude that Theorem \ref{theorem3} holds.
\end{proof}

\subsection{Analysis of the blockchain platform's sub-game}
According to the analysis in Section \ref{subsec:miners}, for any value of reward $R$ given by the blockchain platform, there 
always exists a unique Nash equilibrium for the miners. Therefore, given an value of $R$, the blockchain platform has a unique
utility, and it can maximize its utility by setting an optimal $R$. Substituting the result of Algorithm \ref{algorithm:miner} into
equation (\ref{equ:utility1}) and combining equation (\ref{equ:sum-u}), we have 
\begin{gather}
	\begin{split}
		U &= \alpha \cdot \left[\sigma \left(\beta \cdot \sum\nolimits_{s_i\in S'}\mu_i^*\right) - \frac{1}{2}\right]  - R\\
		&= \alpha \cdot \left[ \sigma \left(\beta \cdot \frac{RN(q-1)}{\sum\nolimits_{s_i\in S'} \lambda_i} \right)- \frac{1}{2}\right] - R \\
		&= \alpha \cdot \left[\sigma \left(\beta X R\right)- \frac{1}{2}\right]  - R,
	\end{split}
\end{gather}
where $X = \frac{N(q-1)}{\sum\nolimits_{s_i\in S'} \lambda_i} $.

\begin{theorem}
	There exists a unique Stackelberg equilibrium $(\bm{\mu^*}, R^*)$ in our proposed Stackelberg game, where $\bm{\mu^*}$
	and $R^*$ are optimal strategies for miners and blockchain platform.
\end{theorem}
\begin{proof}
	As the definition of the blockchain platform's sub-game, the utility function $U(\cdot)$ is defined with $R\in[0, B]$. 
	We then calculate the first order and second order derivatives of $U(\cdot)$ with respect to $R$,
	\begin{gather}
		\frac{\partial U}{\partial R} = \alpha \beta X  \sigma(\beta XR) \left(1 - \sigma(\beta XR)\right) -1 ,
	\end{gather}
	\begin{gather}
		\frac{\partial^2 U}{\partial R^2} = \alpha \beta^2 X^2  \sigma(\beta XR)  (1-\sigma(\beta XR)) 
		 (1 - 2\sigma(\beta XR)).
	\end{gather}
	
	As $\beta XR \geq 0$, the range of the sigmoid function $\sigma(\beta XR)$ is $[\frac{1}{2}, 1)$, and then we have 
	$1 - \sigma(\beta XR) > 0$ and $1 - 2\sigma(\beta XR) \leq 0$. Thus $\frac{\partial^2 U}{\partial R^2} \leq 0$ holds.
	Therefore the utility function $U(\cdot)$ is strictly concave with $R$ for $R \in [0, B]$. It means that a 
	unique $R^*$ can be found to maximize $U(\cdot)$. Combined with Theorem \ref{theorem3}, we conclude that
	there exists a unique Stackelberg equilibrium in our proposed Stackelberg game.	
\end{proof}

The maximization of $U(\cdot)$ is achieved either at the extreme point where the first order derivative of $U(\cdot)$ equals
to $0$, or at the boundary of domain area (i.e., $R = 0 ~or ~B$). 
If $\alpha \beta X < 4$, we have $\frac{\partial U}{\partial R} < 0,~\forall R \in [0, B]$, then $U(\cdot)$ is a decreasing function
in $[0, B]$, and the best strategy of the blockchain platform is $R^*=0$.
If $\alpha \beta X \geq 4$, by solving the equation $\frac{\partial U}{\partial R} = 0$, we have $\sigma(\beta XR) = 
\sqrt{\frac{1}{4} - \frac{1}{\alpha \beta X}} + \frac{1}{2}$, and thus the best strategy of the blockchain platform is 
$R^* = \min \left\{\frac{1}{\beta X} \log \left ( \frac{\frac{1}{2} + \sqrt{\frac{1}{4} - \frac{1}{\alpha \beta X} } }{\frac{1}{2} - 
\sqrt{\frac{1}{4} - \frac{1}{\alpha \beta X} } } \right), B\right\}$.
In summary, we have
\begin{gather}
	R^* = 
	\begin{cases}
		0, ~~~~~~~~~~~~~~~~~~~~~~~~~~~~~~~~~~~~~\text{ if } \alpha \beta X < 4, \\
		\min \left\{\frac{1}{\beta X} \log \left ( \frac{\frac{1}{2} + \sqrt{\frac{1}{4} - \frac{1}{\alpha \beta X} } }{\frac{1}{2} - 
\sqrt{\frac{1}{4} - \frac{1}{\alpha \beta X} } } \right), B\right\}, \text{otherwise}.
	\end{cases}
\end{gather}


\section{Performance Evaluation}\label{sec:Simulation}
In this section, we conduct extensive simulations to evaluate the performance of our proposed incentive mechanism for
blockchain-based internet of things.
\subsection{Experimental settings}
In our experiments, we set the basic parameters of our problem as follows. We assume there are totally $1000$ IIoT devices
that are interested in participating in the blockchain mining process, i.e., $|S| = 1000$. The unit price $\lambda_i$ of each 
miner purchasing computational power from edge servers is uniformly range from $100$ to $105$. 
For the blockchain platform utility model, $\alpha$ is set to be $10000$, $\beta$ is set to be $0.001$, and $B$ is 
set to be $2000$.
We assume that it takes an average of $10$ minutes for the blockchain platform to generate a new block, and thus 
it will generate an average of $144$ new blocks per day, i.e., $N$ is set to be $144$.
Unless otherwise stated, the above parameters will be set as default settings.
Each value in figures in this section is the average of $100$ runs.

\begin{figure} 
	\centering
	\includegraphics[width=.45\textwidth]{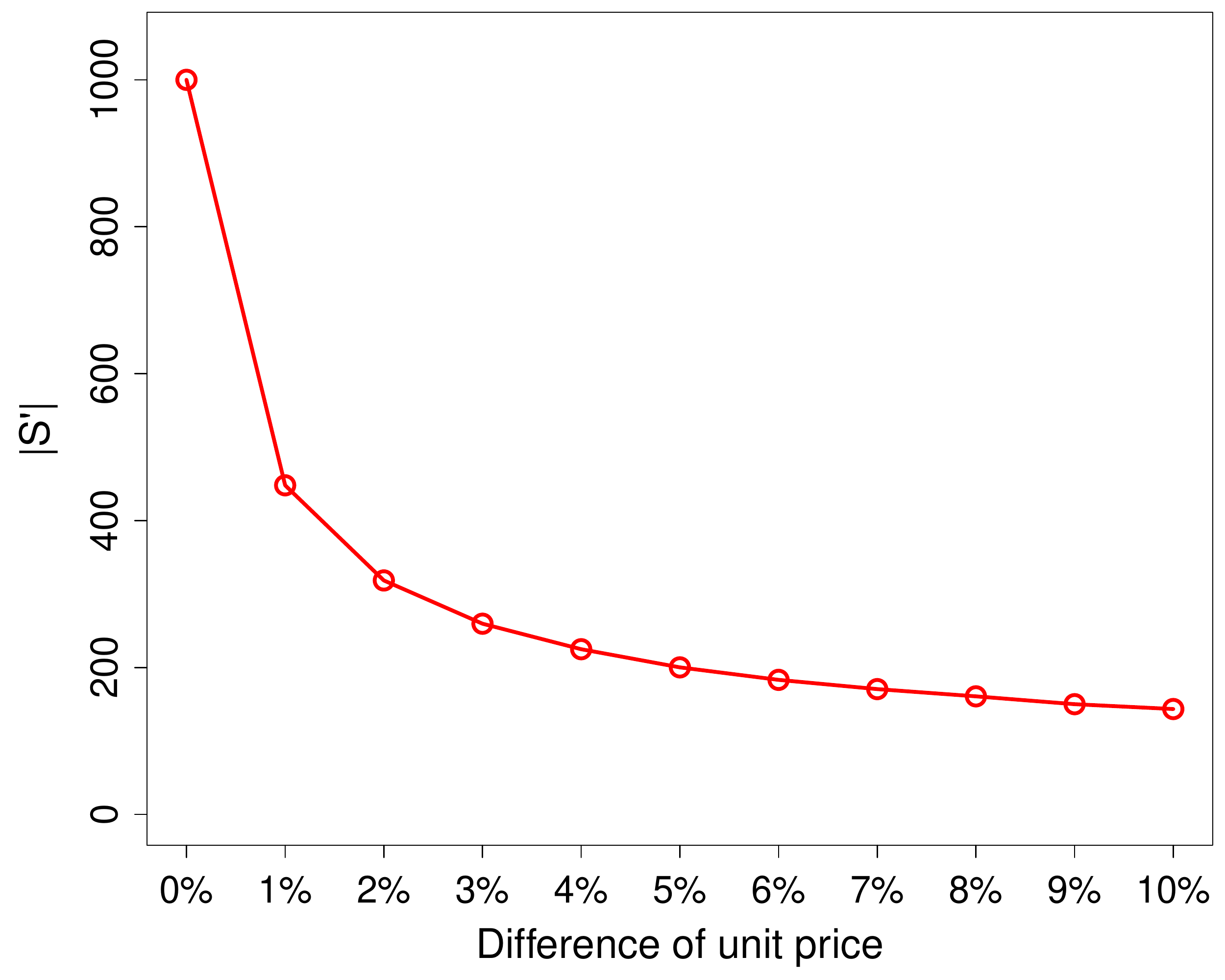}
	\caption{ Impact of the difference of unit price on $|S'|$.}
	\label{fig:nonzero-maxPri}
\end{figure}

\begin{figure} 
	\centering
	\includegraphics[width=.45\textwidth]{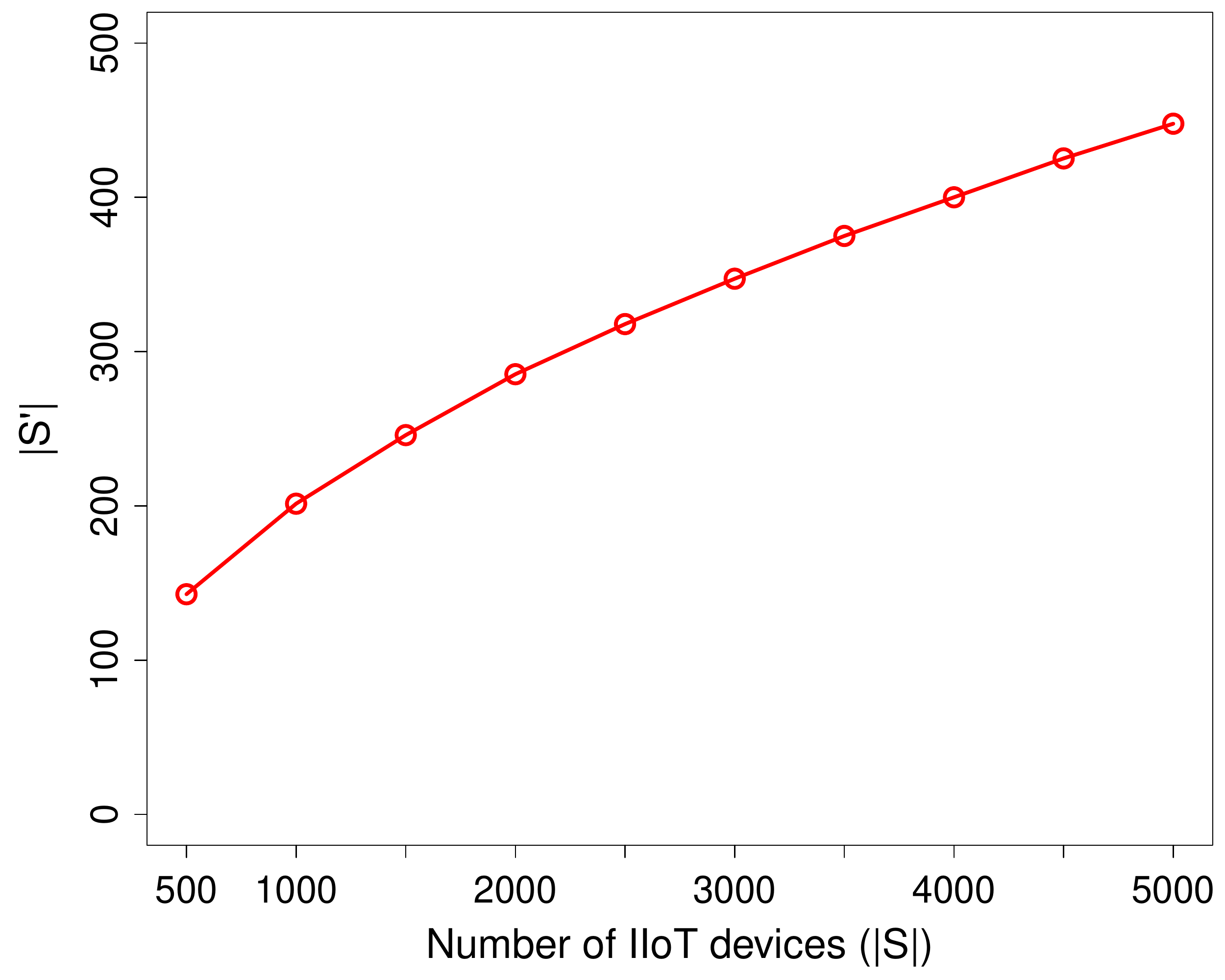}
	\caption{ Impact of number of IIoT devices on $|S'|$.}
	\label{fig:nonzero-minerNum}
\end{figure}

\subsection{Results and Analyses}
\emph{1) Number of participating miners ($|S'|$):} As described in Algorithm \ref{algorithm:miner}, only the miners in $S'$ will 
purchase computational power to participate in the mining process, then we study the impact of unit price of computational
power and number of IIoT devices on $|S'|$.

In Fig. \ref{fig:nonzero-maxPri}, the minimum unit price of computational power is set to be $100$, and $|S|$ is set to be $1000$. 
The difference of unit price represents the range of  $\lambda_i$ for each miner. For example, when the difference of unit 
price is set to be $2\%$, then $\lambda_i$ is randomly chosen in $[100, 102]$ for each miner. 
From Fig. \ref{fig:nonzero-maxPri} we can see that $|S'|$ decreases when the difference of unit price of 
computational power increases. This is because when the difference of unit price is large, $\lambda_i$ of each miner will 
become diverse, and thus there will be more miners violate the condition of Algorithm \ref{algorithm:miner}.
We can also see that the effect of the difference of the unit price of computational power is very significant,
even when the difference of unit price is as small as $1\%$, there are only about $44.8\%$ miners participate the 
blockchain mining process. 

In Fig. \ref{fig:nonzero-minerNum}, we fix the difference of unit price of computational power at $5\%$, and study how the 
number of IIoT devices affect $|S'|$. It can be seen that the growth of $|S'|$ does not have a linear relationship with
the number of IIoT devices. When $|S| = 500$, there are about $28.5\%$ miners participate in the blockchain mining process,
and when $|S|$ increased to $5000$, there are only about $9\%$ miners participate the blockchain mining process.

\begin{figure} 
	\centering
	\includegraphics[width=.45\textwidth]{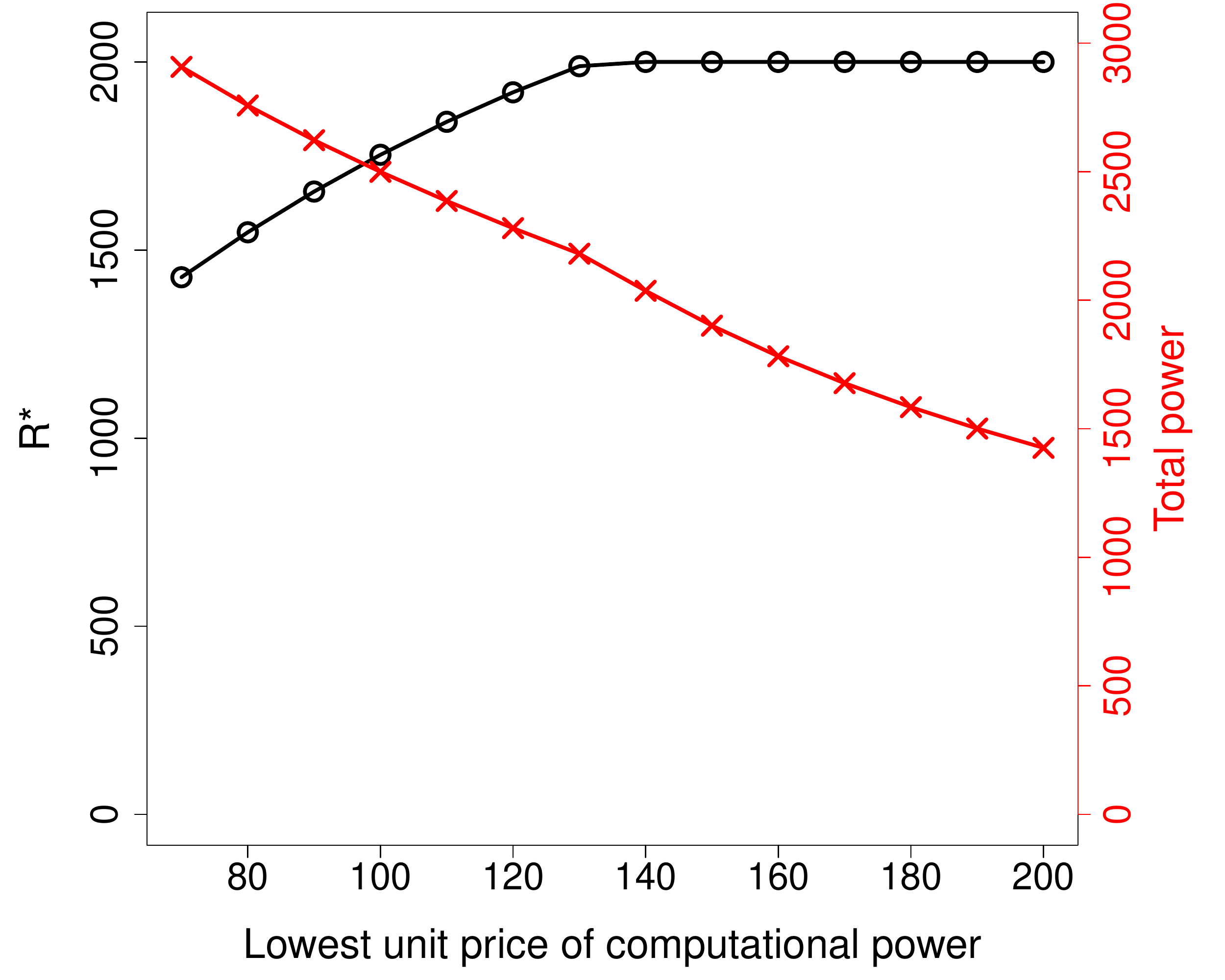}
	\caption{ Impact of the $\lambda_i$ on strategies.}
	\label{fig:unitPri-Strategy}
\end{figure}

\begin{figure} 
	\centering
	\includegraphics[width=.45\textwidth]{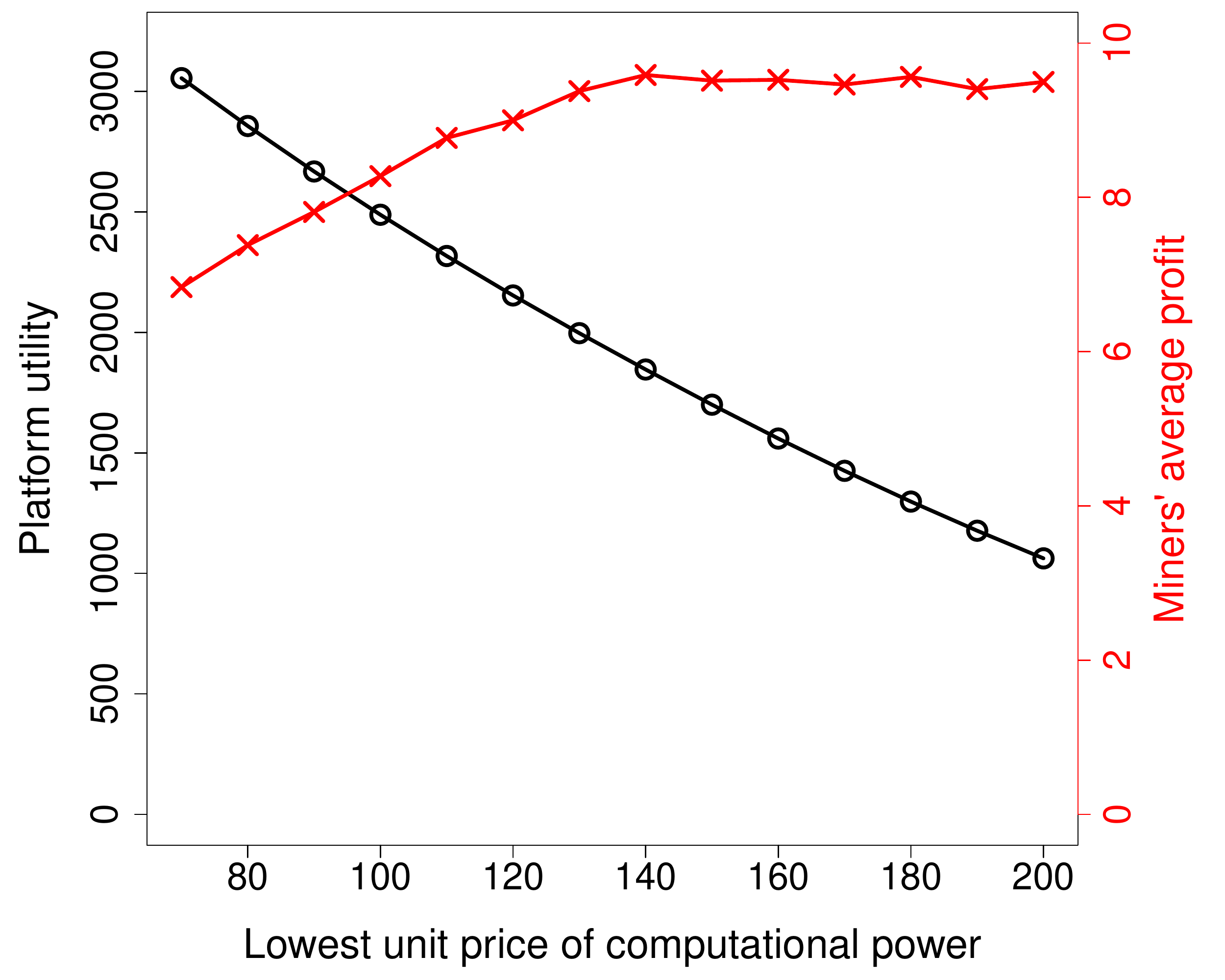}
	\caption{ Impact of $\lambda_i$ on platform utility and miners profits.}
	\label{fig:unitPri-Profit}
\end{figure}

\emph{2) Effect of unit price $\lambda_i$ on utilities and strategies:}
We fix the difference of the unit price of computational power at $5\%$, and study how $\lambda_i$ of each miner affect the 
utilities and strategies of blockchain platform and miners. 

Fig. \ref{fig:unitPri-Strategy} shows that when the unit price of computation power increases, 
the blockchain platform needs to improve the reward to achieve maximum utility until the maximum value of the reward is reached. 
And miners tend to purchase less computational power as it will cost more money. 
Combining Fig. \ref{fig:unitPri-Strategy} and Fig. \ref{fig:unitPri-Profit}, we can observe that the platform utility decreases as the 
unit price of computational power increases, this is because the platform improves the reward but the total computational 
power provided by miners decreases. 
We can also see that the average profit of miners has increased, although they have to pay more money to afford unit 
computational power. The reason is that miners reduce the amount of computational power they purchased, and 
the platform gives more rewards to them.

\begin{figure} 
	\centering
	\includegraphics[width=.45\textwidth]{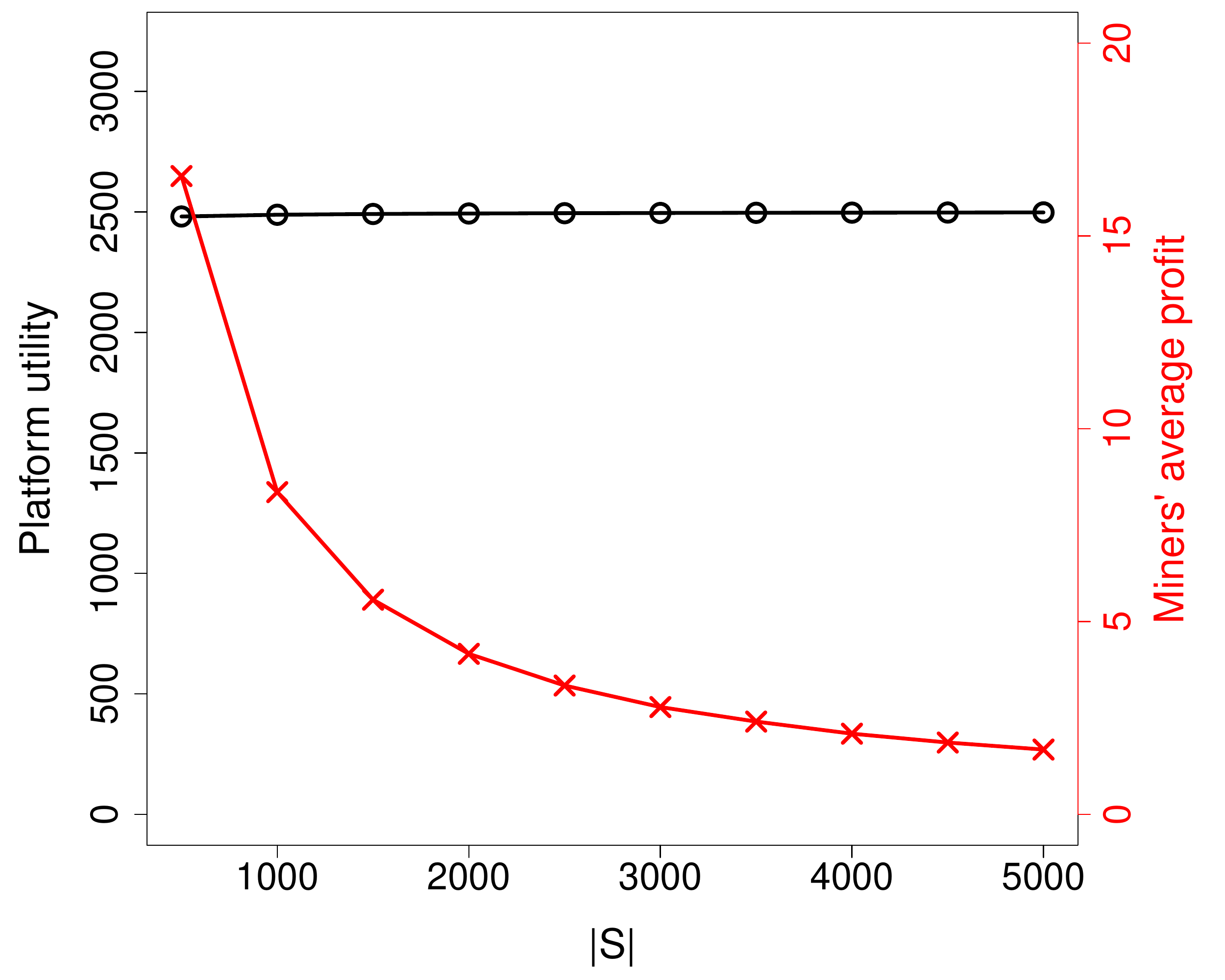}
	\caption{ Impact of $|S|$ on platform utility and miners' profits.}
	\label{fig:minerNum-Profit}
\end{figure}

\begin{figure} 
	\centering
	\includegraphics[width=.45\textwidth]{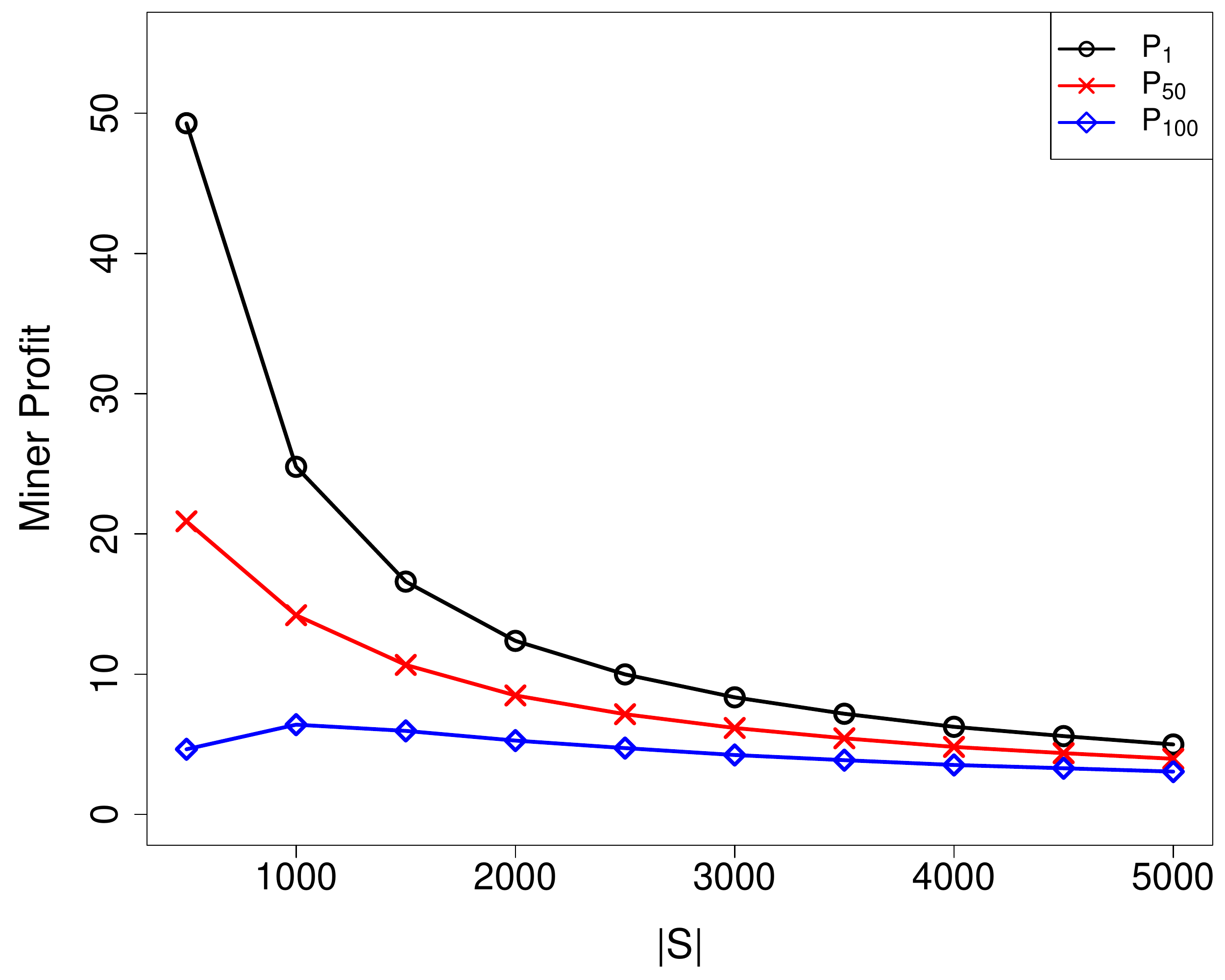}
	\caption{ Impact of $|S|$ on the $1st$, $50th$ and $100th$ miners' profits.}
	\label{fig:minerNum-Profit1-50-100}
\end{figure}

\emph{3) Effect of $|S|$ on platform utility and miners' profits:} 
We fix the range of the unit price of computational power at $[100, 105]$, and study the effect of $|S|$.

As shown in Fig. \ref{fig:minerNum-Profit}, no matter how many IIoT devices are in the blockchain network, the platform always
has a stable utility. However, miners' average profit will decrease as more miners will participate in the mining process and 
intensified the competition. 

Fig. \ref{fig:minerNum-Profit1-50-100} shows the profits of miners $s_1$, $s_{50}$ and $s_{100}$, note that all miners 
have been sorted in ascending order of $\lambda_i$ and renumbered. We can see that the profits of miners $s_1$ and $s_{50}$
decrease with the increases of $|S|$, the result is similar to that in Fig. \ref{fig:minerNum-Profit}. However, for
miner $s_{100}$, its profit increases when $|S|$ increases from $500$ to $1000$, then slowly decreases as $|S|$ increases.
The reason is that when $|S|$ increases from $500$ to $1000$, $\lambda_i$ of each miner becomes more tight, and then the
difference between its unit price of computational power and that of other miners in front of it decreases.
So miner $s_{100}$ become more competitive and thus can get more profits. In detail, when $|S| = 500$, 
$\lambda_1 = 100.011$, $\lambda_{50} = 100.504$ and $\lambda_{100} = 100.992$; when $|S| = 1000$, 
$\lambda_1 = 100.004$, $\lambda_{50} = 100.247$ and $\lambda_{100} = 100.496$.
We can also see that the unit price of computational power will significantly affect the profits of miners, especially when the
number of participating miners is small. For example, when $|S| = 500$, the profit of miner $s_1$ is $2.5$ times the profit
of the miner $s_{50}$ and $10.6$ times the profit of miner $s_{100}$, while the unit price of the computational power of the
miners $s_{50}$ and $s_{100}$ is $0.49\%$ and $0.98\%$ larger than that of $s_1$, respectively.

\section{Conclusion}\label{sec:Conclusions}

In this paper, we design an incentive mechanism for IIoT blockchain network that can be used to motivate IIoT devices to 
purchase more computational resources form edge servers to
participate in the blockchain mining process, thus that a secure blockchain network can be estabilished. 
We analyze the relationship between the security of the blockchain network and the total computational power of the entire 
network, and give the probability that an attacker can successfully tamper with the blockchain.
We model the interaction between blockchain platform and
IIoT devices as a two-stage Stackelberg game, in which the leader, i.e., the blockchain platform, first set the value of the reward 
to maximize its utility,
and miners act as followers, adjust their strategies to maximize their profits. We prove the existence and uniqueness of 
the Stackelberg equilibrium and design an efficient algorithm to compute the Stackelberg equilibrium point. We also 
conduct extensive simulations to evaluate the performance of our designs. Our work is helpful for the IIoT blockchain platform
to set a reasonable reward to build a secure blockchain network.


\ifCLASSOPTIONcaptionsoff
  \newpage
\fi



%
\bibliographystyle{IEEEtran}
\bibliography{paper}

%

\begin{IEEEbiography}[{\includegraphics[width=1in,height=1.25in,clip,keepaspectratio]{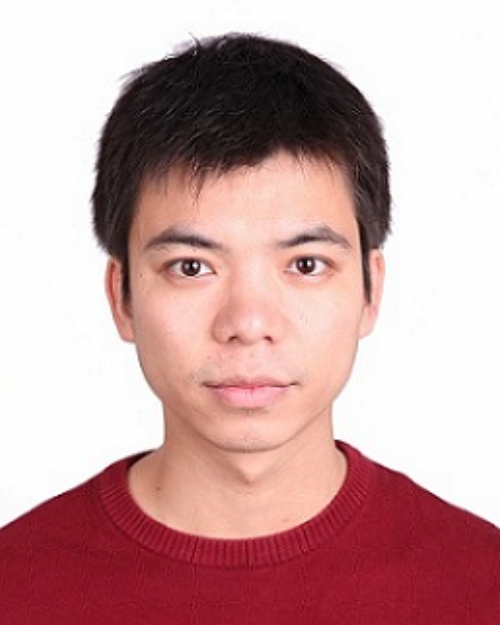}}]{Xingjian Ding}
Xingjian Ding received the BE degree in electronic information engineering from Sichuan University, Sichuan, China, in 2012. He received the M.S. degree in software engineering from Beijing Forestry University, Beijing, China, in 2017. Currently, he is working toward the PhD degree in the School of Information, Renmin University of China, Beijing, China. His research interests include wireless rechargeable sensor networks, algorithm design and analysis, and blockchain.
\end{IEEEbiography}

\begin{IEEEbiography}[{\includegraphics[width=1in,height=1.25in,clip,keepaspectratio]{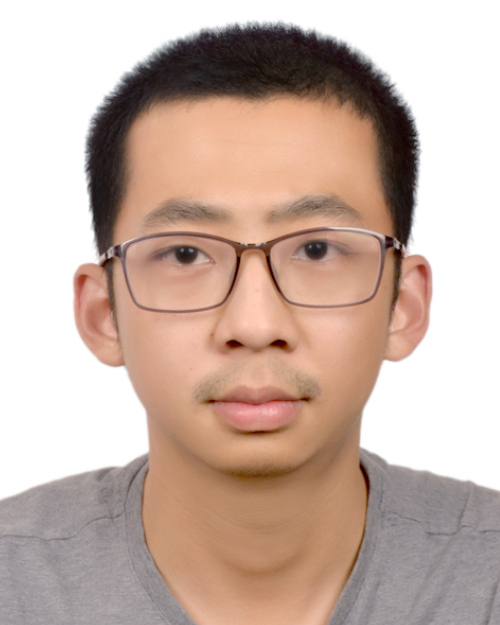}}]{Jianxiong Guo}
Jianxiong Guo is a Ph.D candidate in the Department of Computer Science at the University of Texas at Dallas. He received his BS degree in Energy Engineering and Automation from South China University of Technology in 2015 and MS degree in Chemical Engineering from University of Pittsburgh in 2016. His research interests include social networks, data mining, IoT application, blockchain, and combinatorial optimization.
\end{IEEEbiography}

\begin{IEEEbiography}[{\includegraphics[width=1in,height=1.25in,clip,keepaspectratio]{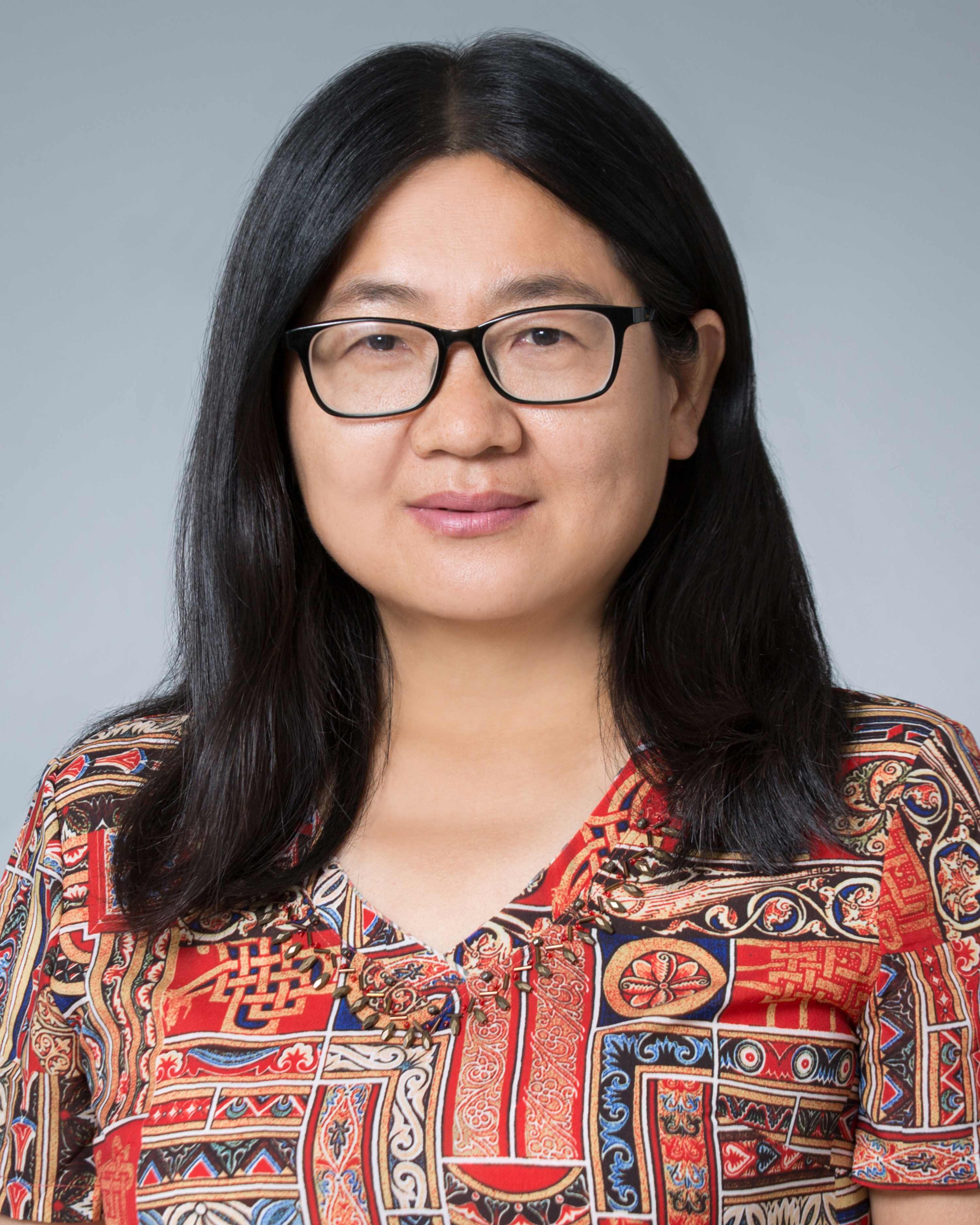}}]{Deying Li}
Deying Li is a professor of Renmin University of China. She received the B.S. degree and M.S. degree in Mathematics from Huazhong Normal University, China, in 1985 and 1988 respectively. She obtained the PhD degree in Computer Science from City University of Hong Kong in 2004. Her research interests include wireless networks, ad hoc \& sensor networks mobile computing, distributed network system, Social Networks, and Algorithm Design etc.
\end{IEEEbiography}

\begin{IEEEbiography}[{\includegraphics[width=1in,height=1.25in,clip,keepaspectratio]{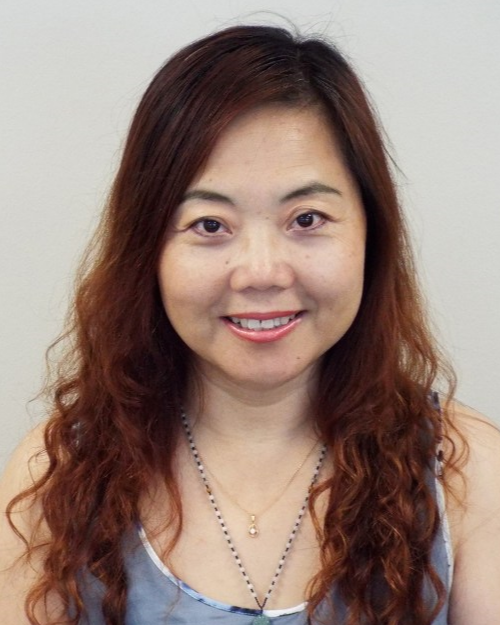}}]{Weili Wu}
Weili Wu received the Ph.D. and M.S. degrees from the Department of Computer Science, University of Minnesota, Minneapolis, MN, USA, in 2002 and 1998, respectively. She is currently a Full Professor with the Department of Computer Science, The University of Texas at Dallas, Richardson, TX, USA. Her research mainly deals in the general research area of data communication and data management. Her research focuses on the design and analysis of algorithms for optimization problems that occur in wireless networking environments and various database systems.
\end{IEEEbiography}





\end{document}